\documentclass{article}
\usepackage{amsmath}
\usepackage{bm}
\usepackage{amsthm}
\usepackage{amssymb}
\usepackage{stmaryrd}
\usepackage[colorlinks=true]{hyperref}
\newtheorem{proposition}{Proposition} 
\newtheorem{definition}{Definition} 
\newtheorem{theorem}{Theorem} 
\newtheorem{corollary}{Corollary} 
\newtheorem{lemma}{Lemma} 
\newcommand{\interp}[1] {\left\llbracket #1 \right\rrbracket}
\newcommand{\annoted}[3]{\overbrace{#3}^{#2}\left.\vphantom{#3}\right\rbrace{\scriptstyle #1}}

\usepackage{tikz}
\usetikzlibrary{backgrounds, shapes.geometric, decorations.pathreplacing}
\pgfdeclarelayer{edgelayer}
\pgfdeclarelayer{nodelayer}
\pgfsetlayers{background,edgelayer,nodelayer,main}
\tikzstyle{none}=[inner sep=0mm]
\tikzstyle{every loop}=[]

\tikzstyle{gn}=[inner sep=0.5pt,minimum width=.3cm,minimum height=0pt,draw=black,shape=circle,fill=green,font=\footnotesize, align=center]
\tikzstyle{grn}=[inner sep=0.5pt,minimum width=.3cm,minimum height=0pt,draw=black,shape=circle,fill=gray!30,font=\footnotesize, align=center]
\tikzstyle{rn}=[inner sep=0.5pt,minimum width=.3cm,minimum height=0pt,draw=black,shape=circle,fill=red!70,font=\footnotesize, align=center]
\tikzstyle{H box}=[rectangle,fill=yellow,draw=black,xscale=1,yscale=1,font=\small,inner sep=0.75pt,minimum width=0.15cm,minimum height=0.15cm]
\tikzstyle{ug}=[regular polygon, regular polygon sides=3, fill=red,draw=black,inner sep = 0pt,minimum width=1em]
\title{The rational fragment of the ZX-calculus}
\author{Emmanuel Jeandel\\Universit\'e de Lorraine, CNRS, Inria, LORIA, F 54000 Nancy, France}
\begin{document}

\maketitle
\begin{abstract}
We introduce here a new axiomatisation of the rational fragment of the
ZX-calculus, a diagrammatic language for quantum mechanics. Compared
to the previous axiomatisation introduced in \cite{NormalForm}, our
axiomatisation does not use any meta-rule, but relies instead on a
more natural rule, called the cyclotomic supplementarity rule, that
was introduced previously in the literature.

Our axiomatisation is only complete for diagrams using rational
angles, and is not complete in the general case.
Using results on diophantine geometry, we characterize precisely
which diagram equality involving arbitrary angles are provable in our
framework without any new axioms, and we show that our axiomatisation
is continuous, in the sense that a diagram equality involving arbitrary
angles is provable iff it is a limit of diagram equalities involving
rational angles.

We use this result to give a complete characterization of all Euler
equations that are provable in this axiomatisation.
\end{abstract}  

\section*{Introduction}

The ZX-calculus is a graphical calculus for quantum computing
introduced by Coecke and Duncan \cite{interacting}, which represent
quantum evolutions (matrices) by \emph{diagrams}. The language is
powerful enough to represent in particular quantum circuits, and is
equipped with a set of transformation rules for diagrams so that one
may simplify a diagram without changing its interpretation (the matrix
it represents).

The initial set of rules as presented in \cite{interacting} is not
complete, in the sense that there exists two diagrams that represent
the same matrix but such that one could not be transformed into the other
using the set of rules.
Thi initial set of rules was studied and refined over the years
\cite{scalar-completeness,simplified-stabilizer,supplementarity} until some
complete axiomatisation were found for the ZX-calculus
\cite{ZXcomplete,Cliffordcomplete,HFW,NormalForm}.

The exact set of rules one should use depends on the particular
fragment of the ZX-calculus we need, i.e. the set of angles/rotations
one may use inside the diagrams.
Of particular interest are the $\pi/4$ fragment (an approximatively
universal fragment), the rational fragment (where most quantum
circuits, e.g. the Quantum Fourier transform
\cite{interacting,nielsen_chuang_2010}, may be expressed exactly), and the general fragment (where every matrix can be
represented).

In this article, we deal mainly with the rational fragment, for which
a complete axiomatisation was already given in
\cite{NormalForm}. However, the set of rules the authors found to
obtain a complete axiomatisation is arguably not satisfying, as it
either relies on using an additional generator (the triangle
generator), or on using so-called meta-rules. We solve this problem in
this article  by proving that one can use instead the generalized
cyclotomic rule, that was introduced previously in \cite{ZXCycl},
which is an easier rule to state and to use.

We then compare our axiomatisation with the axiomatisation of the full
fragment found in \cite{ZXcomplete}. This full axiomatisation relies
on adding only one new nonlinear rule, but this rule, while necessary,
is arguably hard to use and understand. To better understand  this
new rule, we investigate when it is necessary, by
characterizing exactly which diagram inequalities can be proven
without using this rule. We prove in particular that diagram
equalities provable without this rule are exactly the diagram
equalities that are ``limits'' (in some sense) of diagram equalities
involving only rational angles. Using this result and classical
results on rotations in the three-dimensional space, we will in the
last section give a complete list of all so-called Euler equations that can be
proven without this rule.

\section{Definitions and Statement of the Results}

\subsection{Syntax and Semantics}
Let $G$ be a subgroup of $\mathbb{R}$ that contains $2\pi$.
We call $G$ a \emph{fragment}. Of particular interest will be the
following fragments:
\begin{itemize}
  \item $G = \mathbb{Z}\frac{\pi}{4}$. As an abuse of notation, we will
    call $G$ the $\pi/4$ fragment.
  \item $G = \{ \frac{k\pi}{2^n}, k, n \in \mathbb{Z}\}$. We will
    call $G$ the $\frac{\pi}{2^\bullet}$ fragment.
  \item $G = \mathbb{Q}\pi$. This is the \emph{rational} fragment. We will call
    any angle of $G$ a  \emph{rational} angle.
  \item $G = \mathbb{R}$ This is the \emph{full} fragment (therefore
    not really a fragment).
\end{itemize}

Given a fragment $G$, a ZX diagram $D: k \rightarrow l$ with k inputs
and $l$ outputs is generated by the following arrows:

\begin{center}
\begin{tabular}{|cc|cc|}
\hline
$R_Z^{(n,m)}(\alpha):n\to m$ & \input figs/gn-alpha  &
$R_X^{(n,m)}(\alpha):n\to m$ & \input figs/rn-alpha \\
\hline
$H:1\to 1$ & \input figs/Hadamard  & $e:0\to 0$ & \begin{tikzpicture}[baseline=(current bounding box.center)]
\begin{pgfonlayer}{nodelayer}
\path[draw=,dashed=] (0.25, 0.75) rectangle (0.75, 1.25);
\node at (0.5,   1) {};

\end{pgfonlayer}
\begin{pgfonlayer}{edgelayer}

\end{pgfonlayer}
\end{tikzpicture}
 \\\hline
$\mathbb{I}:1\to 1$ & \begin{tikzpicture}[baseline=(current bounding box.center)]
\begin{pgfonlayer}{nodelayer}
\node [style=none] (0) at (0.5, 0.167) {};
\node [style=none] (1) at (0.5, 0.5) {};
\node [style=none] (2) at (0.5, 0.833) {};

\end{pgfonlayer}
\begin{pgfonlayer}{edgelayer}
\draw[out=270, in= 90] (1.center) to (0.center);
\draw[out=270, in= 90] (2.center) to (1.center);

\end{pgfonlayer}
\end{tikzpicture}
  & $\sigma:2\to 2$ &
\input figs/crossing \\\hline
$\epsilon:2\to 0$ & \input figs/cup  & $\eta:0\to 2$ & \input figs/cap  \\\hline
\end{tabular}

where $n,m\in \mathbb{N}$ and $\alpha \in G$. The generator $e$ is the empty diagram.
\end{center}
and the two compositions:

\begin{itemize}
\item Spacial Composition: for any $D_1:a\to b$ and $D_2:c\to d$, $D_1\otimes D_2:a+c\to b+d$ consists in placing $D_1$ and $D_2$ side by side, $D_2$ on the right of $D_1$.
\item Sequential Composition: for any $D_1:a\to b$ and $D_2:b\to c$, $D_2\circ D_1:a\to c$ consists in placing $D_1$ on the top of $D_2$, connecting the outputs of $D_1$ to the inputs of $D_2$.
\end{itemize}

The standard interpretation of the ZX-diagrams associates to any diagram $D:n\to m$ a linear map $\interp{D}:\mathbb{C}^{2^n}\to\mathbb{C}^{2^m}$ inductively defined as follows:\\
\rule{\columnwidth}{0.5pt}

\[ \interp{D_1\otimes D_2}:=\interp{D_1}\otimes\interp{D_2} \qquad 
\interp{D_2\circ D_1}:=\interp{D_2}\circ\interp{D_1}\]
\[\interp{\begin{tikzpicture}[baseline=(current bounding box.center)]
\begin{pgfonlayer}{nodelayer}
\path[draw=,dashed=] (0.25, 0.75) rectangle (0.75, 1.25);
\node at (0.5,   1) {};

\end{pgfonlayer}
\begin{pgfonlayer}{edgelayer}

\end{pgfonlayer}
\end{tikzpicture}
 }:=\begin{pmatrix}
1
\end{pmatrix} \qquad
\interp{~\begin{tikzpicture}[baseline=(current bounding box.center)]
\begin{pgfonlayer}{nodelayer}
\node [style=none] (0) at (0.5, 0.167) {};
\node [style=none] (1) at (0.5, 0.5) {};
\node [style=none] (2) at (0.5, 0.833) {};

\end{pgfonlayer}
\begin{pgfonlayer}{edgelayer}
\draw[out=270, in= 90] (1.center) to (0.center);
\draw[out=270, in= 90] (2.center) to (1.center);

\end{pgfonlayer}
\end{tikzpicture}
 \ }:= \begin{pmatrix}
1 & 0 \\ 0 & 1\end{pmatrix}\qquad
\interp{~\input figs/Hadamard \ }:= \frac{1}{\sqrt{2}}\begin{pmatrix}1 & 1\\1 & -1\end{pmatrix}\]
$$\interp{\input figs/crossing }:= \begin{pmatrix}
1&0&0&0\\
0&0&1&0\\
0&1&0&0\\
0&0&0&1
\end{pmatrix} \qquad
\interp{\input figs/cup}:= \begin{pmatrix}
1&0&0&1
\end{pmatrix} \qquad
\interp{\input figs/cap}:= \begin{pmatrix}
1\\0\\0\\1
\end{pmatrix}$$

For any $\alpha\in G$, $\interp{\begin{tikzpicture}[xscale=0.5,yscale=0.4,baseline=(current bounding box.center)]
\begin{pgfonlayer}{nodelayer}
\node [style=gn,label={[label position=right]$\alpha$}] (0) at (0.5, 0.5) {};

\end{pgfonlayer}
\begin{pgfonlayer}{edgelayer}

\end{pgfonlayer}
\end{tikzpicture}
 }:=\begin{pmatrix}1+e^{i\alpha}\end{pmatrix}$, and
for any $n,m\geq 0$ such that $n+m>0$:
$$
\interp{\input figs/gn-alpha}:=
\annoted{2^m}{2^n}{\begin{pmatrix}
  1 & 0 & \cdots & 0 & 0 \\
  0 & 0 & \cdots & 0 & 0 \\
  \vdots & \vdots & \ddots & \vdots & \vdots \\
  0 & 0 & \cdots & 0 & 0 \\
  0 & 0 & \cdots & 0 & e^{i\alpha}
 \end{pmatrix}}
$$
\begin{minipage}{\columnwidth}
$$\interp{\input figs/rn-alpha}:=\interp{~\input figs/Hadamard ~}^{\otimes
  m}\circ \interp{\input figs/gn-alpha}\circ \interp{~\input
  figs/Hadamard ~}^{\otimes n}$$ \\
$\left(\text{where }M^{\otimes 0}=\begin{pmatrix}1\end{pmatrix}\text{ and }M^{\otimes k}=M\otimes M^{\otimes k-1}\text{ for any }k\in \mathbb{N}^*\right)$.\\
\rule{\columnwidth}{0.5pt}
\end{minipage}\\

In the particular case where we want to emphasize the phase group $G$
in which the angles $\alpha$ live, we will write $ZX_G$ instead of $ZX$.
\clearpage
\subsection{Rules for the $\frac{\pi}{2^\bullet}$ fragments}

\begin{figure}[htbp]
  \centering
  \makebox[0pt]{
  \begin{tabular}{|c|c|c|}
 \hline
 &&\\
 \framebox{
   \input figs/spider
}
&
 \framebox{
   \input figs/S2
}   
&
 \framebox{
   \input figs/S3
 }
 \\[4ex]
 \framebox{
   \input figs/B1
}   
& 
 \framebox{
   \input figs/B2
}   
 &
 \framebox{ 
   \input figs/E
 }\\[4ex]
\framebox{ 
  \input figs/C
}  
&
\framebox{
  \input figs/EU
}  
&
\framebox{
  \input figs/K
}
\\[4ex]
\framebox{
  \input figs/BW
}
&
\framebox{
  \input figs/H
}  
&
\framebox{
  \input figs/SUP
  }
\\
\hline
\end{tabular}
}
 \caption{Rules}
 \label{fig:dyadic}
\end{figure}  

The set of rules of Figure~\ref{fig:dyadic} has been introduced in
\cite{Cliffordcomplete} and has been proven complete for the $\pi/4$ fragment.
As subtler distinctions are not needed in this article, we will
suppose that all the rules of Figure~\ref{fig:dyadic} are given for
all possible, real, values of the angles $\alpha,\beta,\gamma$.

Given two diagrams $D_1, D_2$, we say that $ZX \vdash D_1 = D_2$ if
we can prove that the two diagrams are equal using only the axioms of
Figure~\ref{fig:dyadic}, along with some natural topological axioms
not presented for simplicity.

The theorem  of~\cite{Cliffordcomplete} can be rephrased as follows: If $D_1$ and
$D_2$ are two diagrams with angles only in the $\pi/4$ fragment, and
$\interp{D_1} = \interp{D_2}$, then $ZX \vdash D_1 = D_2$.
The result is actually a bit stronger, as we only to consider the
restrictions of the rules of Figure~\ref{fig:dyadic} to angles in the
$\pi/4$ fragment (said otherwise: to prove that two diagrams of the
$\pi/4$ fragment are equal, we do not need e.g. to introduce the angle
$\pi/8$ or any other angle). We write this for reference but we stress
again that this distinction is not needed in this paper.

Later in \cite{NormalForm}, it has been proven that the same set of rules is
actually complete for the $\pi/2^\bullet$ fragment.
In fact, using the proof methods from \cite{ZXcomplete}, one can prove
an even stronger result:
\begin{theorem}
  \label{thm:gendyadic}
  Let $D_1(\mathbf{\alpha})$ and
  $D_2(\mathbf{\alpha})$ be two ZX-diagrams linear in $\mathbf{\alpha} = \alpha_1, \dots
  \alpha_k$ with constants in $\frac{\pi}{2^\bullet} \mathbb{Z}$.

  If $\interp{D_1(\mathbf{\alpha})} = \interp{D_2(\mathbf{\alpha})}$ for
  all values of $\alpha$, then
  $ZX \vdash D_1(\mathbf{\alpha}) = D_2(\mathbf{\alpha})$ for all
  values of $\mathbf{\alpha}$.
\end{theorem}
By a diagram linear in $\mathbf{\alpha}$, we mean that every angle
appearing in the diagram is a linear combination with integer
coefficients of the angles $\alpha_i$, possibly with a constant term
in $\frac{\pi}{2^\bullet} \mathbb{Z}$. The exact definition may be
found in \cite{ZXcomplete}.
While the theorem was only proven for the $\pi/4$ fragment in
\cite{ZXcomplete}, it is easy to see that the same method of proof
gives the result for this larger fragment.

\subsection{The $\mathbb{Q}\pi$-fragment}

The set of axioms of Figure~\ref{fig:dyadic} is  not enough for
completeness of bigger fragments of the ZX-calculus and  additional
axioms are needed.

In \cite{NormalForm}, the authors have introduced a ``meta-axiom'' to obtain a
set of rules complete for the $\mathbb{Q}\pi$ fragment.

If one does not want meta-axioms, one can see from the same article that the set of axioms represented
in Figure~\ref{fig:cyc} is enough to obtain a complete set of rules:
If two diagrams $D_1, D_2$ with rational angles are equal, then
$ZX \cup \{ CYC_p, p \geq 3\} \vdash D_1 = D_2$.

\begin{figure}[h]
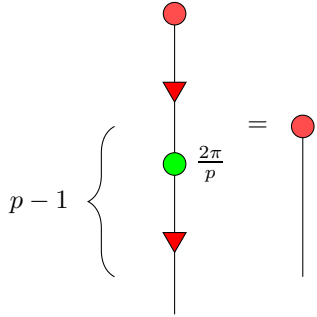

  \input figs/def-cyc
  \caption{The cyclotomic rules $(CYC_p)$.}
  \label{fig:cyc}
\end{figure}  

This set of rules relies on an additional generator, the triangle
generator, which is related to the right side of the $(BW)$ rule:
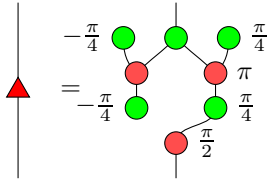
\begin{center}
\input figs/triangle
\end{center}
The standard interpretation of this new generator is the matrix
$\begin{pmatrix} 1 & 1 \\ 0 & 1\end{pmatrix}$.

\subsection{The general fragment}

Adding the cyclotomic axioms is again not enough to obtain
completeness for the general fragment. An additional axiom, presented
in Figure~\ref{fig:general} was introduced in \cite{ZXcomplete}, and
it is proven that this new axiom is sufficient to obtain a complete
axiomatisation (Another axiomatisation, relying on much more axioms, was also given
independently in \cite{HFW}).

\begin{figure}[h]
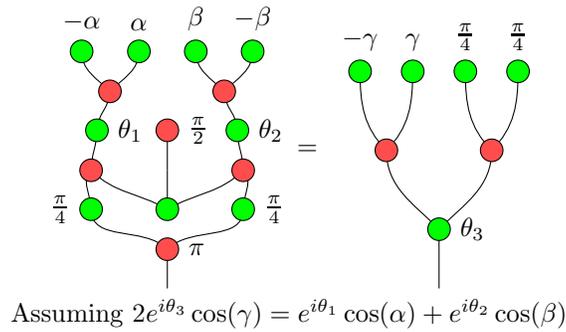

\begin{center}  
  \input figs/general
  
Assuming $2e^{i\theta_3} \cos(\gamma) = e^{i\theta_1}\cos(\alpha) +
e^{i\theta_2} \cos(\beta)$
\end{center}
  \caption{The general rule $(A)$.}
  \label{fig:general}
\end{figure}

Compared to the other rules, it is important to note that this one
cannot be applied everywhen, but only if the angles satisfy the
technical condition stated on the rule.

\subsection{Results}

We are now ready to state the main results of the paper.

The results are twofold. First, we prove that the ugly rule of
Figure~\ref{fig:cyc} can be represented by the well known
generalized supplementarity rule.
This rule was introduced in \cite{ZXCycl} and is the content of
Figure~\ref{fig:sup}.

\begin{figure}[h]
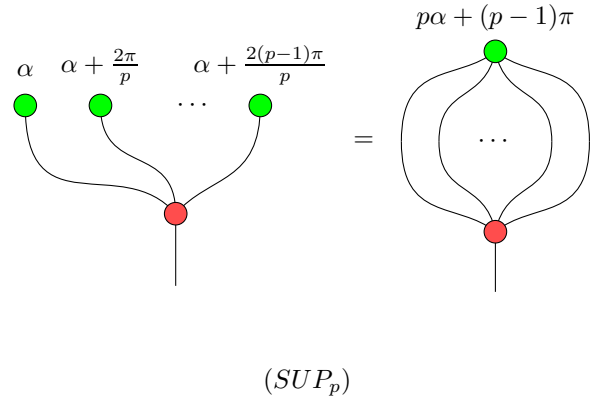

  \input figs/def-sup
  \caption{The generalized supplementarity rules $(SUP_p)$.}
  \label{fig:sup}
\end{figure}  

\noindent The result may be stated as follows: For any prime number $p$, ${ZX
\cup \{SUP_p\} \vdash CYC_p}$. It is interesting to note that the proof
uses complicated angles: To prove $CYC_{11}$ for instance, our proof
uses diagrams involving the angle $\pi/32$.

Our second series of results concern the other ugly rule of
Figure~\ref{fig:general}. This rule is the only rule that is nonlinear
and involves a precondition to be tested. This rule is known to be
necessary, but we are interested here in what can be proved without
this rule.

Here are our results:
\begin{itemize}
  \item An equality $D_1 = D_2$ is provable without rule $(A)$
    iff it is a specialization of a more general diagram 
    equality $D'_1(\alpha_1, \alpha_2 \dots, \alpha_n) =D'_2(\alpha_1,
    \alpha_2, \dots \alpha_n)$ with rational angles, that is:
    \begin{itemize}
      \item For all values of $\alpha_i$, $ZX \vdash D'_1(\alpha_1, \alpha_2
        \dots, \alpha_n) =D'_2(\alpha_1,     \alpha_2, \dots
        \alpha_n)$
      \item $D'_1$ and $D'_2$ use, apart from $\alpha_1 \dots
        \alpha_n$,  only rational angles.
      \item For some value of $\alpha_1 \dots \alpha_n$,
        $D_1$ and $D_2$ are equal respectively to $D'_1(\alpha_1,
        \alpha_2 \dots \alpha_n)$ and $D'_2(\alpha_1,
        \alpha_2 \dots \alpha_n)$.
    \end{itemize}

  \item
    An equality $D_1 = D_2$ is provable without rule $(A)$
    iff it is a ``limit'' of diagram equalities involving only rational
    angles: There exists two sequences of diagrams $D_1^n$ and $D_2^n$
    with rational angles s.t.
    
    \begin{itemize}
    \item For all $n$, $ZX \vdash D_1^n  = D_2^n$
    \item $D_1^n \xrightarrow{n \to\infty} D_1$
    \item $D_2^n \xrightarrow{n \to\infty} D_2$
    \end{itemize}      

    Convergence here is to be understood as point-wise convergence of
    angles: Diagrams $D_1^n$ and $D_2^n$ are structurally identical to
    $D_1$ and $D_2$, but with different angles, that  converge to the
    real angles when $n$ goes to infinity.
  \item As a consequence of the previous result, we give a complete
    characterization of all equalities related to the Euler
    decomposition that are provable without rule $(A)$. This list
    is not surprising and contains all usual suspects.
\end{itemize}

\section{Cyclotomic Supplementarity as a sufficient rule}

In this section, we will prove that the ZX-calculus as presented in
Figure~\ref{fig:dyadic} together with the supplementarity rules $(SUP_p)$ of
Figure~\ref{fig:sup} is a complete axiomatisation of the rational
fragment of the ZX-calculus. 

The article \cite{NormalForm} introduced a meta-rule to make the ZX-calculus
complete for the rational fragment. However, it is also clear from the
article that the set of rules $CYC_p$ gives also a complete
axiomatisation:
\begin{theorem}[\cite{NormalForm}]
  The axioms of ZX, as presented in Figure~\ref{fig:dyadic}, together with
  the axioms $(CYC_p)$ for prime numbers $p$ as presented in Figure~\ref{fig:sup} is a complete axiomatisation of the rational fragment of the ZX-calculus.
\end{theorem}

We will prove in this section that we can replace the cyclotomic
axioms $CYC_p$ by the supplementarity axioms:

\begin{theorem}
  Let $p$ be an odd prime number. Then
$ZX \cup \{SUP_p\} \vdash CYC_p$.
\end{theorem}

We now proceed to the proof of the theorem. In all the following, $p$ is fixed to some odd prime number.

We will first explain informally the structure of the proof.
If terms of interpretation, the supplementarity rule $SUP_p$ is stating that for all $\alpha$, 
\[
\begin{pmatrix}
  P_1(e^{i\alpha},e^{i\beta})\\
  P_2(e^{i\alpha},e^{i\beta}) 
\end{pmatrix}   =
\begin{pmatrix}
  P_3(e^{i\alpha},e^{i\beta})\\
  P_4(e^{i\alpha},e^{i\beta}) 
\end{pmatrix} 
\]
where $P,Q,P',Q'$ are some polynomials, and $\beta = 2\pi/p$.

The cyclotomic rule $CYC_p$ is stating that
\[
\begin{pmatrix}
1\\
R(e^{i\beta}) 
\end{pmatrix}   =
\begin{pmatrix}
  1\\
  0\\
\end{pmatrix} 
\]
for some polynomial $R$ and $\beta = 2\pi/p$.

To obtain $CYC_p$ from $SUP_p$, we will proceed in two steps:
\begin{itemize}
  \item First, rewrite $SUP_p$ in such a way that the new diagram
    equality can be interpreted, as saying, in terms of interpretation
    that:
\[
\begin{pmatrix}
1\\
S(e^{i\alpha},e^{i\beta}) 
\end{pmatrix}   =
\begin{pmatrix}
  1\\
  0\\
\end{pmatrix} 
\]
when $\beta = 2\pi/p$.
   \item It turns out that the polynomial $R$ involved in the
     cyclotomic rule $CYC_p$ is actually one of the terms of the
     polynomial $S$. In the next step, we will explain how this term
     can be extracted diagrammatically, which will prove the result.    
\end{itemize}

Now, for the first step. Using elementary manipulation, we get to rewrite the supplementarity in the form
\begin{lemma}
  The following diagram equation is provable in $ZX \cup \{SUP_p\}$:

  \input figs/fig-sup
\end{lemma}

In terms of matrices, this can be interpreted as the following
equality:

\[
  \frac{1}{2}
\begin{pmatrix}
(1+e^{ip\alpha}) \prod_k (e^{i\alpha+2k\pi/p}+1) + (1-e^{ip\alpha})
\prod_k (e^{i\alpha+2k\pi/p}-1)\\
-(1-e^{ip\alpha}) \prod_k (e^{i\alpha+2k\pi/p}+1) - (1+e^{ip\alpha}) \prod_k (e^{i\alpha+2k\pi/p}-1)    
\end{pmatrix}
=
\begin{pmatrix}
2e^{pi\alpha}\\
  0
\end{pmatrix}
\]

or more simply

\[
  \begin{pmatrix}
    P(e^{i\alpha})\\
    Q(e^{i\alpha})\\    
  \end{pmatrix}
  = 
\begin{pmatrix}
 2e^{pi\alpha}\\
  0
\end{pmatrix}
\]
for some polynomials $P$ and $Q$. The next step is to split this
diagram into two diagrams, one that speaks of $P$, and one that speaks
of $Q$.

For this we introduce three diagrams. $D$ is the generalization of the
supplementarity diagrams, and $D_1$, $D_2$ somehow corresponds to the
two polynomials.

\begin{definition}
We define three (families of) diagrams as follows:
  
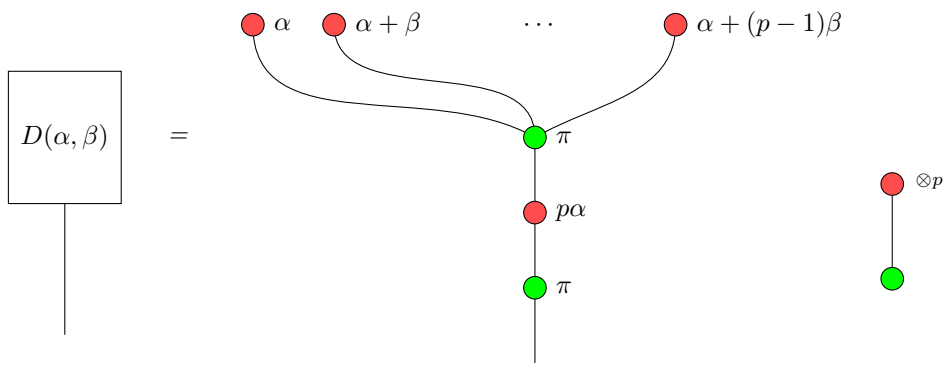
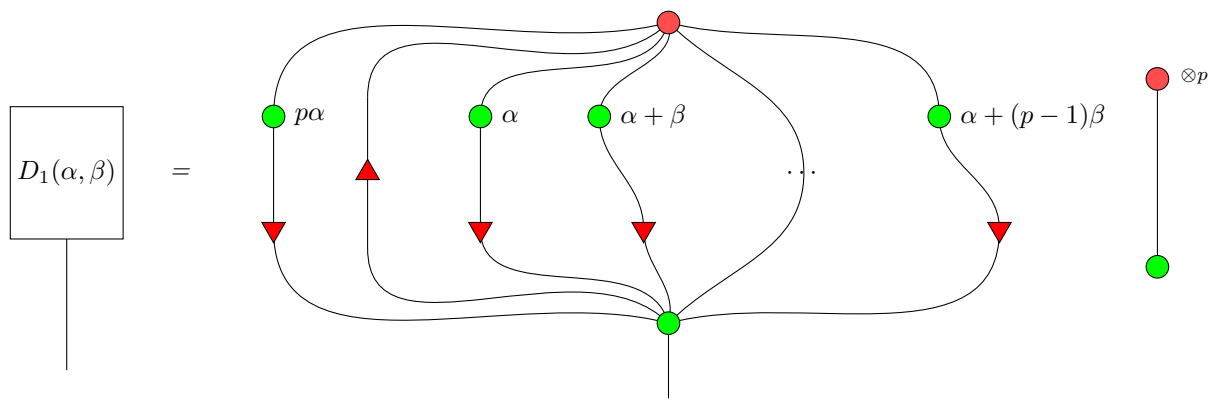
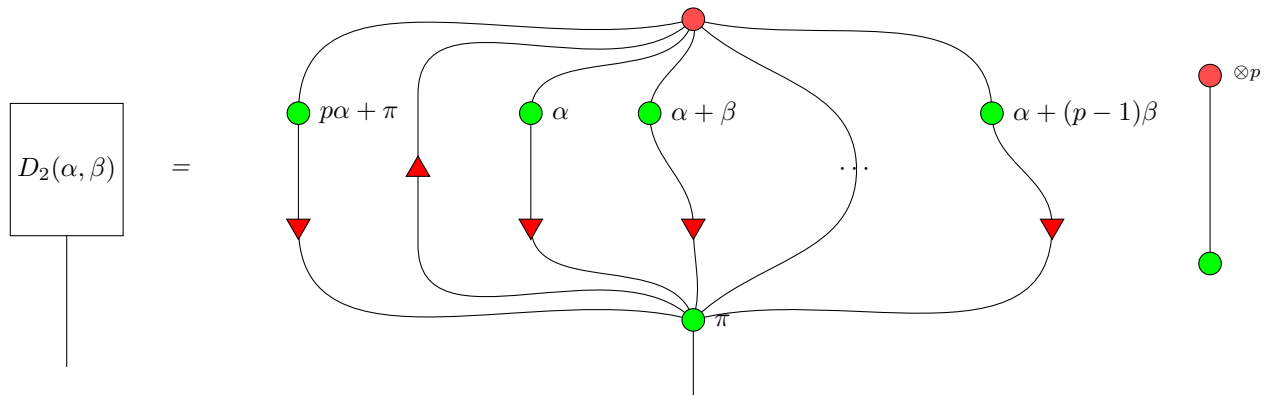
\begin{center}
  \input figs/fig-sup2
\end{center}

\begin{center}
  \input figs/fig-sup3
\end{center}

\begin{center}
  \input figs/fig-sup4
\end{center}
\end{definition}

The three diagrams have the following interpretations:
\[
\interp{D(\alpha,\beta)} =  \frac{1}{2}
\begin{pmatrix}
(1+e^{ip\alpha}) \prod_k (e^{i\alpha+k\beta}+1) + (1-e^{ip\alpha})
\prod_k (e^{i\alpha+k\beta}-1)\\
-(1-e^{ip\alpha}) \prod_k (e^{i\alpha+k\beta}+1) - (1+e^{ip\alpha})
\prod_k (e^{i\alpha+k\beta}-1)
\end{pmatrix} =
  \begin{pmatrix}
    P(e^{i\alpha},e^{i\beta})\\
    Q(e^{i\alpha},e^{i\beta})
  \end{pmatrix}    
\]

\[
\interp{D_1(\alpha,\beta)} = 
\begin{pmatrix}
1\\
\frac{1}{2}\left[
(1+e^{ip\alpha}) \prod_k (e^{i\alpha+k\beta}+1) + (1-e^{ip\alpha})
\prod_k (e^{i\alpha+k\beta}-1)\right]\\
\end{pmatrix}
=
\begin{pmatrix}
  1 \\
  P(e^{i\alpha},e^{i\beta})\\
\end{pmatrix}  
\]

\[
\interp{D_2(\alpha,\beta)} = 
\begin{pmatrix}
1\\
\frac{-1}{2}\left[
(1-e^{ip\alpha}) \prod_k (e^{i\alpha+k\beta}+1) + (1+e^{ip\alpha})
\prod_k (e^{i\alpha+k\beta}-1)\right]\\
\end{pmatrix}
=
\begin{pmatrix}
  1 \\
  Q(e^{i\alpha},e^{i\beta})\\
\end{pmatrix}  
\]

We therefore have the following:

\begin{proposition}
  The following equations are provable in $ZX \cup \{SUP_p\}$:
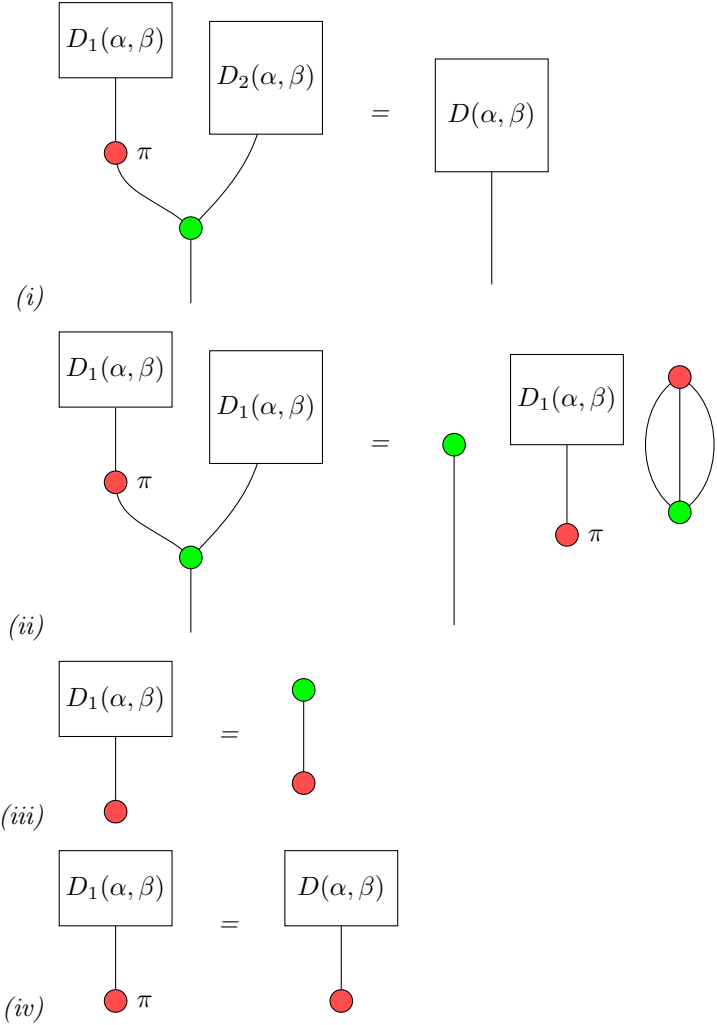
  \begin{enumerate}
\renewcommand\labelenumi{(\roman{enumi})}
\renewcommand\theenumi{$\labelenumi$}
\item   \label{enum:i}
  \input figs/fig-sup5
\item\label{enum:ii}
  \input figs/fig-sup6
\item\label{enum:iii}
  \input figs/fig-sup7
\item\label{enum:iv}
  \input figs/fig-sup7b
\item\label{enum:v}
  \input figs/fig-sup8
\end{enumerate}  
\end{proposition}  
\begin{proof}
  The first four are true for all $\alpha,\beta$ and are therefore
  provable by Theorem~\ref{thm:gendyadic}.

  The last one is the supplementary equality in the form provided by
  the previous lemma.
\end{proof}

\begin{proposition}
  \label{prop:first}
  The following equation is provable in $ZX \cup \{SUP_p\}$:
\begin{center}
  \input figs/fig-sup9
\end{center}  
\end{proposition}

\begin{proof}
By putting \ref{enum:i} and \ref{enum:v} together we get:
\begin{center}
  \input figs/fig-sup10
\end{center}  

Therefore:
\begin{center}
  \input figs/fig-sup11
\end{center}  

Applying \ref{enum:ii} on the left and simplifying the right term:
\begin{center}
  \input figs/fig-sup12
\end{center}  

Now we apply \ref{enum:iv} then \ref{enum:v} on the left and
\ref{enum:iii} on the right:

\begin{center}
  \input figs/fig-sup13
\end{center}  
which gives the result.
\end{proof}

Now the advantage of the diagram $D_2$ over the original
supplementarity equation is that its interpretation is
\[\begin{pmatrix}
  1 \\
  Q(e^{i\alpha},e^{i\frac{2\pi}{p}})\\
\end{pmatrix}  
\]
That is, at least one term of the matrix is fixed to a value that
doesn't depend on $e^{i\alpha}$ (and is nonzero).

Now, the equality of the previous proposition means that $Q(e^{i\alpha},
e^{i\frac{2\pi}{p}})$ is identically $0$, which means that each coefficient of the
polynomial (in $e^{i\alpha}$) is equal to $0$.

It is easy to see that the first coefficient of $Q(e^{i\alpha},e^{i\beta})$
is $0$ (independently of $\beta$).

Notice that the term of degree $1$ (in $e^{i\alpha}$) of the polynomial
$Q(e^{i\alpha}, e^{i\beta})$ is exactly
$-1-e^{i\beta}-e^{i2\beta}-\dots -e^{i(p-1) \beta}$.

If we find a way to retrieve the term of degree $1$ of a polynomial
with a diagrammatic construction, we are close to our result.

This is done with the following result:

\begin{proposition}
  Let $P(X) = \sum_{r \leq d} a_r X^r$ be a polynomial of degree $d$.
  Choose $n$ such that $2^n > d $.

  Then

  \[ \sum_{k = 0}^{2^n-1} \frac{P(e^{2ki\pi/2^n})}{e^{2ki\pi/2^n}} =
    2^n a_1 \]

\end{proposition}

\begin{proof}
  Indeed  $\sum_k e^{2irk\pi/2^n }  = 0$ if $r \not= 0 \mod 2^n$  and
  $\sum_k e^{2irk\pi/2^n }  = 2^n$ otherwise. 
\end{proof}

We will use this idea to extract diagrammatically the first coefficient of our
polynomial.

\begin{definition}
  We define recursively a family of diagrams $W_n$ as follows:

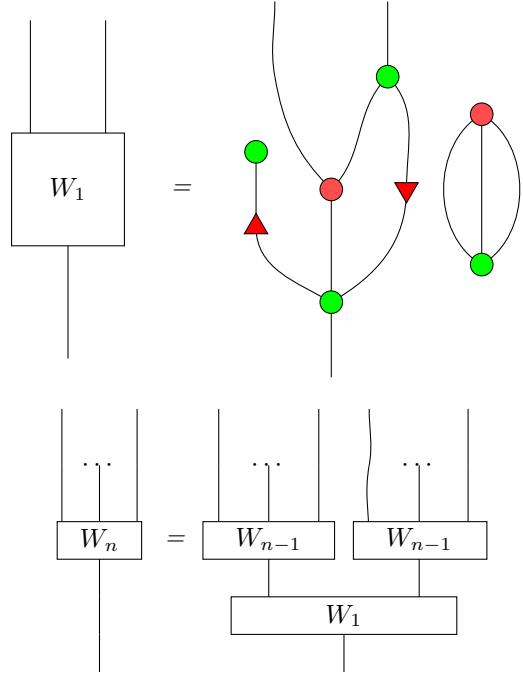
\begin{center}
  \input figs/fig-sup14
\end{center}

\begin{center}
  \input figs/fig-sup15
\end{center}  
\end{definition}

For reference, here are the interpretations of $W_1$ and $W_2$:

\[\begin{pmatrix}
    1 & 0 & 0 & 0 \\
    0 & 1/2 & 1/2 & 0 \\
\end{pmatrix}  \]

\[
\left(  \begin{array}{cccccccccccccccc}
  1&  0&  0&  0&  0&  0&  0&  0&  0&  0&  0&  0&  0&  0&  0&  0\\
  0&1/4&1/4&  0&1/4&  0&  0&  0&1/4&  0&  0&  0&  0&  0&  0&  0\\
        \end{array}
        \right)
\]

It is easy to see by induction that $W_n$ is the matrix whose first
row contains a $1$ in position $0$ and whose second row contains
$2^{-n}$ in position $2^i$ for all $i$.

The astute reader may remark that  $W_n$ is a slight variant of the
W-state presented e.g. in the ZW calculus \cite{zw}.

The previous proposition then gives us immediately the following:
\clearpage
\begin{proposition}
Let $n$ be big enough s.t. $2^{n} > 2p$.
  
  The following equality is provable in $ZX$:

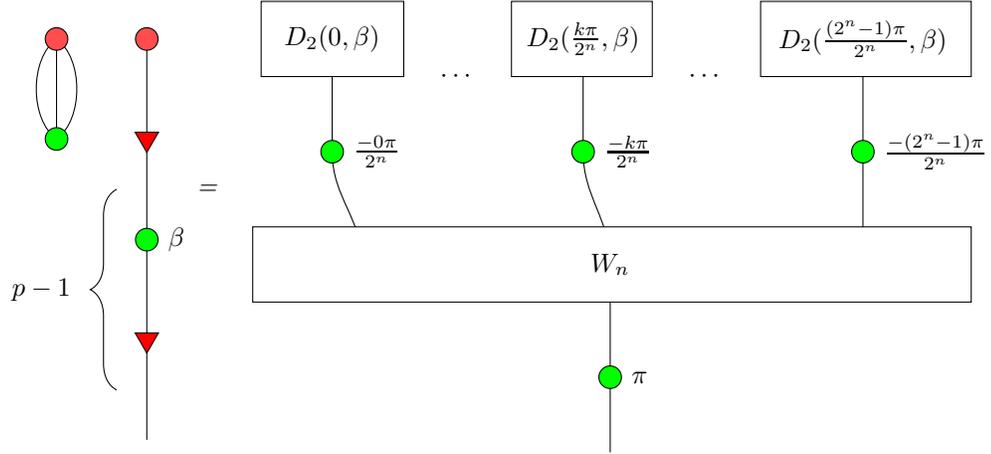
\begin{center}
  \input figs/fig-sup16
\end{center}  
\end{proposition}
(the brace means this part of the diagram should be repeated $p-1$ times).
\begin{proof}
  By the previous proposition, the two diagrams have the same
  interpretations (the interpretation of $D_2$ is a polynomial of
  degree $2p < 2^n$ in $e^{i\alpha}$).

  Now these two diagrams are (apart from $\beta$) diagrams where all
  angles are of the form $k\pi/2^n$ and the two  diagrams have the
  same interpretations for all possible values of $\beta$. Therefore
  the equality is provable by Theorem~\ref{thm:gendyadic}.
\end{proof}

\begin{proposition}
  The following is provable in $ZX \cup \{ SUP_p \}$.
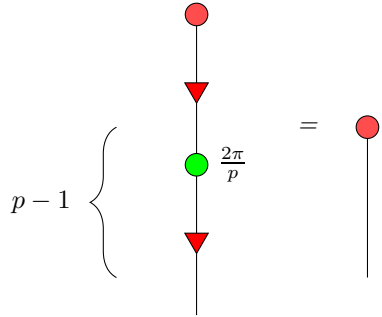
\begin{center}
  \input figs/fig-sup17
\end{center}
That is, $ZX \cup \{ SUP_p \} \vdash CYC_p$.
\end{proposition}

\begin{proof}
  By the previous proposition and Proposition~\ref{prop:first}:

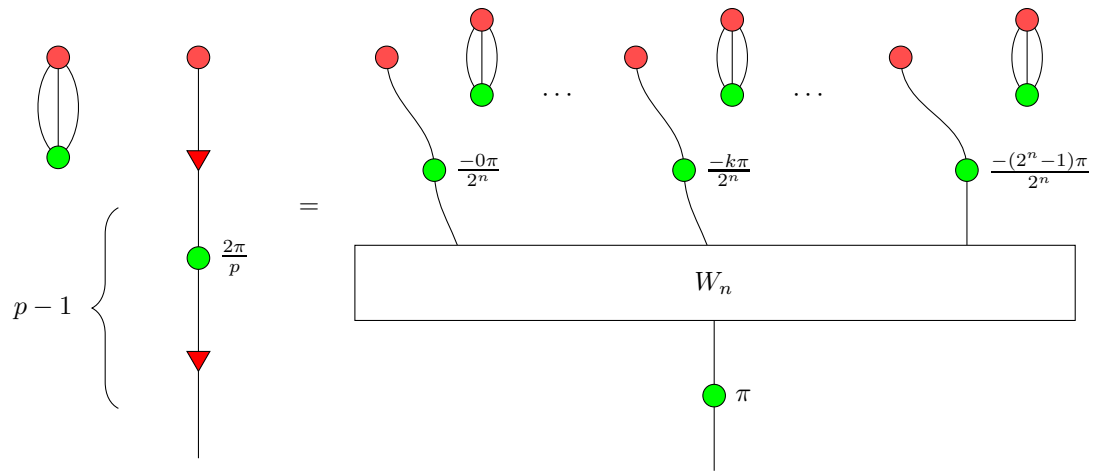
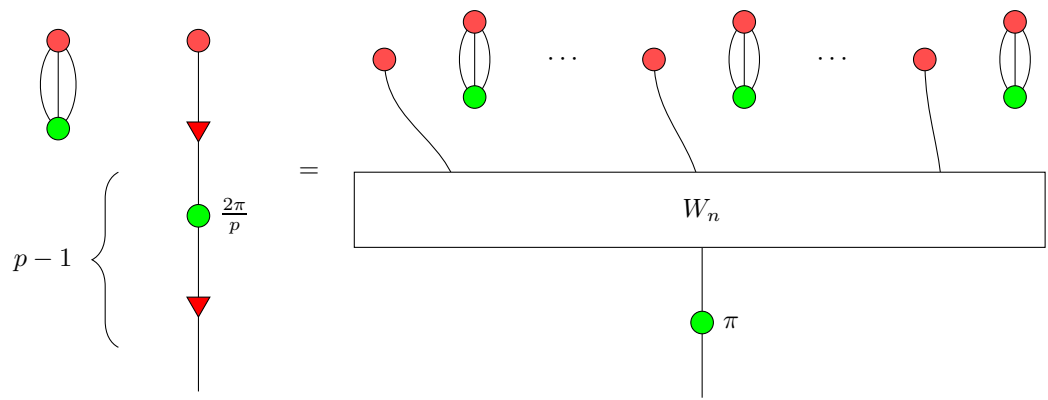
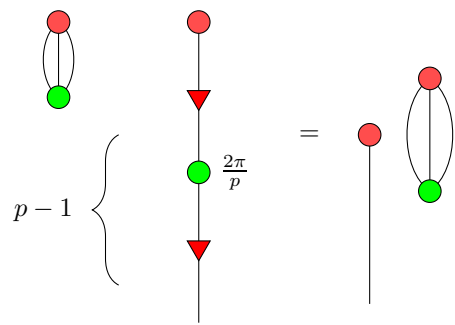
  \begin{center}
    \input figs/fig-sup18
  \end{center}    

Therefore:
  \begin{center}
    \input figs/fig-sup19
  \end{center}    
  And therefore:
  
  \begin{center}
    \input figs/fig-sup20
  \end{center}    
where the last equality comes from the fact the two right terms have
the same interpretations and are in the fragment $\pi/4$ and therefore
are provably equal by completeness.
\end{proof}

\section{Equalities provable without rule $(A)$}
In this section we characterize which diagram equalities   $D_1 = D_2$
can be proven without rule $(A)$.

This section will primarily deal with vectors of variables. 
The notation $\overline{a}$ in this section will denote the vector
$a_1\dots a_n$. The size of the vector is usually either irrelevant or
obvious from the context. On some occasions, we will use the notation
$(\overline{a}_i)_{i \in I}$ which is not a typo, but refers to a
tuple of vectors.

Greek letters will represent variables.
By  an abstract diagram $D_1(\overline{\alpha})$ we mean a  diagram of the ZX-calculus
where some of the angles inside the green and red nodes are replaced
by some of the variables $\alpha_i$.
We say the diagram is rational if the other angles are in $\mathbb{Q}\pi$.
If $\overline{x}$ is a vector of real numbers of the same size as
$\overline{\alpha}$, we denote by $D_1(\overline{x})$ the diagram
where the variable $\alpha_i$ is replaced by the real number $x_i$.

Many of the results of this section are based on the following remark.
Consider two abstract diagrams $D_1(\overline{\alpha})$ and
$D_2(\overline{\alpha})$.
Then the set
\[
S = \{ \overline{x} \in \mathbb{R}^l | \interp{D_1(\overline{x})} = \interp{D_2(\overline{x})}\}
\]
has a special structure that can be exploited.
This  special structure is best evidenced by not seeing $S$ as a
subset of (vectors of) reals, but as points of the unit circle, i.e. by
investigating:
\[
S' = \{ e^{i\overline{x}} | \interp{D_1(\overline{x})} = \interp{D_2(\overline{x})}\}
\]
We will then rely on classical results from algebraic group theory and
diophantine geometry.
The two main results  we will
use is the fact that any linear commutative compact groups is a
quasi-torus \cite{OV} (which is well-known in linear algebraic group
theory, or in Lie theory, and easy to prove once all relevant
definitions are given), and the Mordell-Lang conjecture on rational
points on algebraic varieties (which is lesser known and whose proof
is much harder, even in our context \cite{Laurent}).

To simplify the exposition, we will state these classical results only when
necessary, and only in terms of real numbers (or in terms of numbers
in the quotient space $\mathbb{R}/2\pi\mathbb{Z}$ which is compact),
but the reader should be aware that the results in the cited
litterature are usually expressed in terms of points of the unit
circle.

\paragraph{}
To begin our investigation, we note that we can do the same trick as in the beginning of
the paper: From the fact $ZX \cup \{SUP_p, p \geq 3\}$ is complete for
the rational fragment, we deduce that all \emph{linear} equalities are also
provable.

It will be useful in what follows to state this theorem precisely,
which needs a few notations.

By an affine form, we mean a function of the form
$f(\overline{\alpha}) = f(\alpha_1, \dots, \alpha_k) = p_0 + \sum p_i \alpha_i$ where all
coefficients $p_i$ are integer. The form is \emph{rational} if $p_0$
is a rational angle. It is linear if $p_0 = 0$.
By a (rational) affine transformation $F$,  we mean a finite sequence
of (rational) affine  forms  $f_1 \dots f_n$, and we will write
$F(\overline{\alpha})$ for $f_1(\overline{\alpha}), \dots
f_n(\overline{\alpha})$. Notice that a linear affine  transformation
is essentially a matrix with integer coefficients.

\begin{theorem}\label{thm:paramrat}
  Let $D_1(\overline{\alpha})$ and
  $D_2(\overline{\beta})$ be two abstract rational  ZX-diagrams.

  Let $F,G$ be two rational transformations s.t. for all (vector of) angles $\overline{x}$,

  \[
  \interp{D_1(F(\overline{x}))} = \interp{D_2(G(\overline{x}))}
  \]

  Then for all angles $\overline{x}$,

  \[ ZX \cup \{SUP_p, p \geq 3\} \vdash
  D_1(F(\overline{x})) = D_2(G(\overline{x}))\]
\end{theorem}  
This theorem follows mutatis mutandis from the similar theorem in \cite{ZXcomplete}.

Our first result in this section is that this previous theorem is somehow an if and only if:
\renewcommand\labelenumi{(\roman{enumi})}
\renewcommand\theenumi\labelenumi
\begin{theorem}  
Let $D_1(\overline{\alpha})$ and $D_2 (\overline{\beta})$ be two
abstract rational diagrams of the $ZX$-calculus

Let $\overline{x}$ and $\overline{y}$ be two (vectors of) angles.

Then the following are equivalent:
\begin{enumerate}
  
  \item\label{thm:1} $ZX \cup \{ SUP_p, p \geq 3\} \vdash D_1 (\overline{x}) =
    D_2(\overline{y})$
  \item\label{thm:2} There exists $n$ s.t. for all positive integers $k \equiv 1 \mod n $:
    \[ \interp{D_1 (k\overline{x})} =
    \interp{D_2(k\overline{y})}    \]
  \item\label{thm:3} There exist rational affine transformations $F,G$  s.t.
\begin{itemize}
\item For all values of $\overline{z}$, we have
  \[ \interp{D_1(F(\overline{z}))} = \interp{D_2(G(\overline{z}))} \]
\item There exist some value of $\overline{z}$ s.t.
  $\overline{x} = F(\overline{z})$ and $\overline{y} = G(\overline{z})$
\end{itemize}    
\end{enumerate}  
\end{theorem}

What this theorem means is that there is no coincidence: If a diagram
equality is provable without the rule (A), then the diagram equality
is not provable only for the angles that appear in it, but \ref{thm:2} is
also provable if all angles are multiplied by some constants and \ref{thm:3}
the diagram equality  is actually an instance of a much more general
diagram equality.

We note that it is likely that $\ref{thm:1}\rightarrow \ref{thm:3}$
could be proven directly  by a cumbersome induction on the structure of the
proof. However, we will adopt here an algebraic approach that will be
useful for the next theorem.

\begin{proof}

$\ref{thm:3}\rightarrow \ref{thm:1}$ is exactly the statement of the
previous theorem.

Now let's prove $\ref{thm:1}\rightarrow\ref{thm:2}$.

First, as a diagram with variables  $\overline{\alpha}$ is a fortiori
a diagram with variables  $\overline{\alpha}$ and $\overline{\beta}$,
we may suppose wlog that $D_1$ and $D_2$ are two diagrams with the
same set of variables $\overline{\alpha}$ and that $\overline{x} = \overline{y}$.

Now suppose that $ZX \cup \{ SUP_p, p \geq 3\} \vdash D_1(\overline{x})
= D_2(\overline{x})$. As a proof is finite, only finitely many
axioms of the form $SUP_p$ are used. Let $q$ be an upper bound of the
denominators $p$ of the rational angles $k\pi/p$ that are involved in
the diagrams 
$D_1$ and $D_2$ and on the numbers $p$ s.t.  $SUP_p$ is used in the
proof.

In particular, we have
$ZX \cup \{ SUP_p, p < q\} \vdash D_1(\overline{x}) = D_2(\overline{x})$.

  Let $n = 8q!$ and let $k$ be an integer.Consider the two diagrams $D_1(\overline{x})$ and
  $D_2(\overline{x})$ where all angles are multiplied by $kn+1$.

  By our choice of $q$, it is easy to see that all rational coefficients inside
  $D_1$ and $D_2$ do not change.
  Therefore the two new diagrams we obtain are actually equal to
  $D_1((kn+1)\overline{x})$ and $D_2((kn+1)\overline{x})$

  Furthermore, by our choice of $q$, all axioms of $ZX$ and $SUP_p, p
  < q$ remain valid
  where their angles are multiplied by $kn+1$. As a consequence, we have for all~$k$, 
  
  \[ ZX \cup \{ SUP_p, p < q\} \vdash D_1((kn+1)\overline{x}) = D_2((kn+1)\overline{x})\]

In particular  for all positive values of $k$,      $\interp{D_1((kn+1)\overline{x})} = \interp{D_2((kn+1)\overline{x})}$.
  
Now let's prove $\ref{thm:2}\rightarrow\ref{thm:3}$.
As before we may suppose wlog that $D_1$ and $D_2$ are two diagrams with the
same set of variables $\overline{\alpha}$ and that $\overline{x} =
\overline{y}$. Let $m$ be the size of $\overline{x}$.

Let $S = \{ (kn+1)\overline{x} , k \geq 0\}$, seen as a subset of
$(\mathbb{R}/2\pi \mathbb{Z})^m$ and $X$ be its topological closure. Each
individual coefficient of the interpretation of $D_1(\overline{\alpha})$ and $D_2(\overline{\alpha})$ is an
exponential polynomial in $\overline{\alpha}$, and therefore is continuous. As a
consequence, for all $\overline{z} \in X$, we have
$\interp{D_1(\overline{z})} = \interp{D_2(\overline{z})}$.

It remains to know what $X$ looks like. Such sets have been extensively studied
in the context of algebraic groups. To be more exact,  $\{ kn\overline{x}, k \geq 0\}$ is
a monoid, and its closure $Y$ is a compact subgroup of  $(\mathbb{R}/2\pi
\mathbb{Z})^m$, and each such group is a linear quasi-torus\cite{OV}, which
means that there exists a rational linear transformation $F$ and
finitely many (vectors of) rational numbers $(\overline{l}_i)_{i\in I}$ s.t.

\[ Y = \bigcup_{i \in I} \{ F (\overline{z}) + \overline{l}_i, \overline{z} \in (\mathbb{R}/2\pi \mathbb{Z})^m \} \]

It is a routine exercise\footnote{Notice that $nx \in Y$ therefore $nx = F(\overline{w}) +
\overline{l}_i$ for some $\overline{w}$ and some $i$ thus $x = F(\overline{w}/n) + \overline{\ell}$ for some
 (vector of) rational angles $\overline{\ell}$. The result follows by a change of variables}
 to show that this implies the same result for
$X = x + Y$:
\[ X = \bigcup_{i \in I} \{ F (\overline{z}) + l'_i, \overline{z} \in
(\mathbb{R}/2\pi \mathbb{Z})^m \} = \bigcup_{i \in I} \{ F_i
(\overline{z}) , \overline{z} \in (\mathbb{R}/2\pi \mathbb{Z})^m \} \]
for some rational forms $F_i$.

As a nontrivial example, let $\overline{x} = (\sqrt{2}, 2\sqrt{3},
\pi/3+\sqrt{2}, \sqrt{2}+\sqrt{3})$ and $n = 3$.
Then one can prove with some work (using linear independence of $\sqrt{2}$, $\sqrt{3}$
and $\pi$) that the topological closure $Y$ of $\{ kn\overline{x}, k \geq
0\}$ can be defined by  ${Y = \{ (t, 2u, t,u+t), (t,u) \in
(\mathbb{R}/2\pi \mathbb{Z})^m \} \cup \{ (t, 2u, \pi+t, t+u), (t,u) \in
(\mathbb{R}/2\pi \mathbb{Z})^m \}} $. Similarly, one can prove that the topological
closure $X$ of $S = \{ kn\overline{x}+1, k \geq 0\}$ is
${X = \{ (t, 2u, \pi/3+t,u+t), (t,u) \in
(\mathbb{R}/2\pi \mathbb{Z})^m \} \cup \{ (t, 2u, 4\pi/3 + t, t+u), (t,u) \in
  (\mathbb{R}/2\pi \mathbb{Z})^m \}} $.

Now $\overline{x} \in X$. Therefore there exists $i$ s.t.  $\overline{x} \in \{ F_i (\overline{z}), \overline{z} \in (\mathbb{R}/2\pi
\mathbb{Z})^m \}$. Then

\begin{itemize}
  \item $\{ F_i (\overline{z}), \overline{z} \in (\mathbb{R}/2\pi
    \mathbb{Z})^m \} \subseteq X$ and therefore
$\interp{D_1(F(\overline{z}))} = \interp{D_2(F(\overline{z}))}$ for
    all values of $\overline{z}$.
  \item $x \in \{ F_i (\overline{z}), \overline{z} \in (\mathbb{R}/2\pi\mathbb{Z})^m \}$,
    therefore there exists $\overline{z}$ s.t $x = F_i(\overline{z})$.
\end{itemize}  
\end{proof}

\clearpage
\begin{corollary}
Let $D_1(\overline{\alpha})$ and $D_2 (\overline{\beta})$ be two
rational diagrams of the $ZX$-calculus

Suppose that $ZX \cup \{ SUP_p, p \geq 3\} \vdash D_1 (\overline{x}) = D_2(\overline{y})$ for some
angles $\overline{x}, \overline{y}$.

Then there exist two sequences of (vectors of) rational angles $\overline{x^i}, \overline{y^i}$ that converges
respectively to $\overline{x}$ and $\overline{y}$ s.t.
for all values of $i$, 
\[ ZX \cup  \{ SUP_p , p \geq 3\} \vdash D_1(\overline{x^i}) = D_2(\overline{x^i})\]
  
\end{corollary}  
\begin{proof}
  Use $\ref{thm:1}\rightarrow\ref{thm:3}$ in the previous theorem, and choose a sequence of rational angles $\overline{z}^i$ that converges to
  $\overline{z}$ and take $\overline{x}^i = F(\overline{z}^i)$ and $\overline{y}^i = G(\overline{z}^i)$.
\end{proof}  

This is actually a characterization:
\begin{theorem}
Let $D_1(\overline{\alpha})$ and $D_2 (\overline{\beta})$ be two
rational diagrams of the $ZX$-calculus.

Let $\overline{x^i}, \overline{y^i}$ be two sequences of rational angles that converge
respectively to $\overline{x}$ and $\overline{y}$ and s.t.
for all values of $i$, 
\[ \interp{D_1(\overline{x^i})} = \interp{D_2(\overline{y^i})}\]

Then
\[ ZX \cup  \{ SUP_p , p \geq 3\} \vdash D_1(\overline{x}) = D_2(\overline{y})\]
\end{theorem}

\begin{corollary}
  A limit of provable equalities is again a provable equality.
\end{corollary}

\begin{proof}[Proof ot the Corollary]
Let $\overline{x^i}, \overline{y^i}$ be two sequences of angles that converge
respectively to $\overline{x}$ and $\overline{y}$ and s.t.
for all values of $i$, 
\[ ZX \cup  \{ SUP_p , p \geq 3\} \vdash D_1(\overline{x^i}) = D_2(\overline{y^i})\]

By the previous corollary, we may find, for each $i$,  \emph{rational} angles
arbitrarily close to $\overline{x^i}$ and $\overline{y^i}$  s.t. the equality remains
provable for these angles.

But this implies that we can find rational angles arbitrarility close
to $\overline{x}$ and $\overline{y}$  s.t.  the equality remains
provable for these angles.
Which implies by the previous theorem that $ZX \cup  \{ SUP_p , p \geq
3\} \vdash
D_1(\overline{x}) = D_2(\overline{y})$
\end{proof}

\begin{proof}[Proof of the Theorem]
  As before, by seeing $D_1$ and $D_2$ as diagrams over a bigger set of
  variables, we may suppose wlog that $\overline{x^i} = \overline{y^i}$.

  Now, the key idea is to consider
  \[
X = \{ \overline{z} \in (2\pi\mathbb{Q}/2\pi
\mathbb{Z})^m| \interp{D_1(\overline{z})} = \interp{D_2(\overline{z})} \}
\]
$X$ is therefore the set of  (vectors of) rational angles for which the
diagram equality is true.
We note that all $\overline{x^i}$  are in  $X$.

Now take
\[ Y =   \{ e^{i\overline{z}}, z \in X\} \]

Each individual coefficient of $D_1$ and $D_2$ are polynomials in $e^{iz_1}
\dots e^{iz_l}$.
The set $Y$ itself can therefore be seen as the intersection between the set of
zeroes of some polynomials, and the set of (vectors consisting of) roots of unity.

Such sets have been widely used in the context of the Mordell-Lang conjecture,
and it is known that each of them is the finite union of translates of algebraic
groups, which is in our particular case a result of Laurent \cite{Laurent}.
In our formalism, this can be rewritten as follows: There exists 
finitely many rational affine transformations
$(F_j)_{j \in J}$ s.t. (as before) 

\[ X = \bigcup_{j \in J} \{ F_j (\overline{z}), \overline{z} \in (2\pi\mathbb{Q}/2\pi \mathbb{Z})^m \} \]

Now as $J$ is finite,  infinitely many of the $\overline{x^i}$ are in the same set, say
$F_j$. 
Write $S = \{  F_j (\overline{z}), \overline{z} \in
(2\pi\mathbb{Q}/2\pi \mathbb{Z})^m \}$.
By taking a subsequence, we may suppose that all of the $\overline{x}^i$ are in $S$.

\begin{itemize}
\item $S \subseteq X$. Therefore for all $\overline{w} \in S$,  $\interp{D_1(\overline{w})} =
  \interp{D_2(\overline{w})}$.
\item  Therefore for all $\overline{z}\in   (2\pi\mathbb{Q}/2\pi
  \mathbb{Z})^m$, $\interp{D_1(F_j(\overline{z}))} =
  \interp{D_2(F_j(\overline{z}))}$. By continuity this is also true for any
  $\overline{z} \in (\mathbb{R}/2\pi\mathbb{Z})^m$.  
\item Therefore by Theorem~\ref{thm:paramrat}, the equality is
  not only true but provable: For all $\overline{z}$, we have
  $ZX \cup \{ SUP_p, p\geq 3\} \vdash D_1(F_j(\overline{z})) =
  D_2(F_j(\overline{z}))$
\item Now it remains to show that there exists some $\overline{z}$ s.t.
  $F_j(\overline{z}) = \overline{x}$ which will prove that
$ZX \cup \{ SUP_p, p\geq 3\} \vdash D_1(\overline{x})) =
  D_2(\overline{x})$. But this is obvious: For each $i$, as $\overline{x^i} \in S$, we may find $\overline{z^i}  \in (\mathbb{R}/2\pi\mathbb{Z})^m$
  s.t.   $F_j(\overline{z}^i) = \overline{x}^i$ and we may suppose (upto a
  converging subsequence) that $\overline{z}^i$  converges to some
  $\overline{z}$. And for this particular $\overline{z}$ we have $F_j(\overline{z}) = \overline{x}$.
\end{itemize}  
\end{proof}

\section{Applications to Euler transforms}
As an application of the previous result, we will characterize in this section
the set of Euler equalities  that are provable without rule $(A)$.

By a Euler equality, we mean an equality of the form:
\newcommand\euler[6]{
\begin{tikzpicture}
\begin{pgfonlayer}{nodelayer}
\node [style=none] (0) at (1, 1) {};
\node [style=gn,label={[label position=left]#3}] (1) at (1, 1.5) {};
\node [style=rn,label={[label position=left]#2}] (2) at (1, 2) {};
\node [style=gn,label={[label position=left]#1}] (3) at (1, 2.5) {};
\node [style=none] (4) at (1, 3) {};
\node [style=none] (eq) at (1.5, 2) {$=$};
\node [style=none] (a0) at (2, 1) {};
\node [style=rn,label={[label position=right]#6}] (a1) at (2, 1.5) {};
\node [style=gn,label={[label position=right]#5}] (a2) at (2, 2) {};
\node [style=rn,label={[label position=right]#4}] (a3) at (2, 2.5) {};
\node [style=none] (a4) at (2, 3) {};
\end{pgfonlayer}
\begin{pgfonlayer}{edgelayer}
\draw[out=270, in= 90] (1.center) to (0.center);
\draw[out=270, in= 90] (2.center) to (1.center);
\draw[out=270, in= 90] (3.center) to (2.center);
\draw[out=270, in= 90] (4.center) to (3.center);
\draw[out=270, in= 90] (a1.center) to (a0.center);
\draw[out=270, in= 90] (a2.center) to (a1.center);
\draw[out=270, in= 90] (a3.center) to (a2.center);
\draw[out=270, in= 90] (a4.center) to (a3.center);
\end{pgfonlayer}
\end{tikzpicture}
}

\begin{center}
  \euler{$\alpha_1$}{$\alpha_2$}{$\alpha_3$}{$\beta_1$}{$\beta_2$}{$\beta_3$}
\end{center}
possibly also involving (nonzero) scalars. For simplicity of the presentation, we will not
represent the scalars in the proofs.

Our main result is the following:

\begin{theorem}
The only Euler equalities that are provable in $ZX \cup \{ SUP_p, p
\geq 3\}$ are:

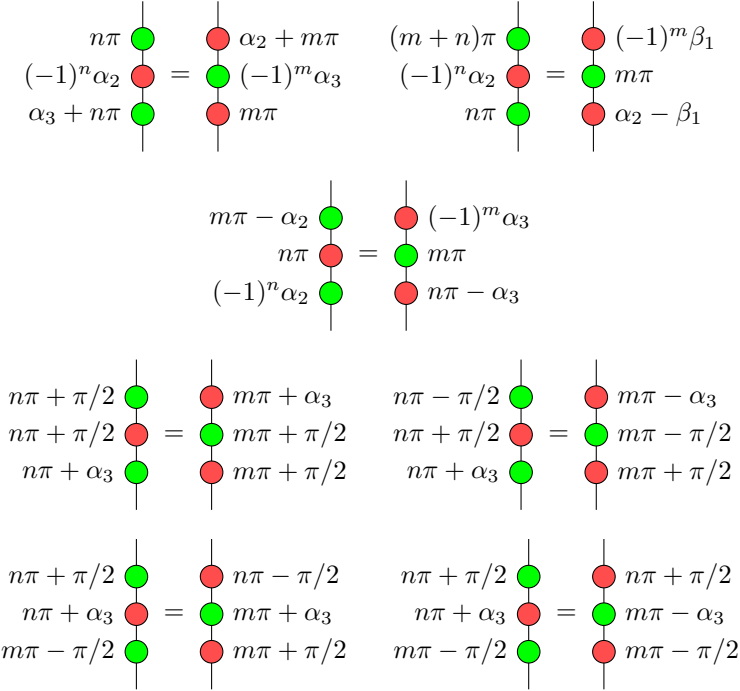
\begin{center}
\euler{$n\pi$}{$(-1)^n\alpha_2$}{$\alpha_3+n\pi$}{$\alpha_2+m\pi$}{$(-1)^m \alpha_3$}{$m\pi$}
\euler{$(m+n)\pi$}{$(-1)^n\alpha_2$}{$n\pi$}{$(-1)^m\beta_1$}{$m\pi$}{$\alpha_2-\beta_1$}
\end{center}
\begin{center}
\euler{$m\pi-\alpha_2$}{$n\pi$}{$(-1)^n \alpha_2$}{$(-1)^m
  \alpha_3$}{$m\pi$}{$n\pi-\alpha_3$}
\end{center}
\begin{center}
\euler{$n\pi+\pi/2$}{$n\pi+\pi/2$}{$n\pi+\alpha_3$}{$m\pi+\alpha_3$}{$m\pi+ \pi/2$}{$m\pi+\pi/2$}
\euler{$n\pi-\pi/2$}{$n\pi+\pi/2$}{$n\pi+\alpha_3$}{$m\pi-\alpha_3$}{$m\pi-
  \pi/2$}{$m\pi+\pi/2$}
\end{center}
\begin{center}
\euler{$n\pi+\pi/2$}{$n\pi+\alpha_3$}{$m\pi-\pi/2$}{$n\pi-\pi/2$}{$m\pi+\alpha_3$}{$m\pi+
  \pi/2$}
\euler{$n\pi+\pi/2$}{$n\pi+\alpha_3$}{$m\pi-\pi/2$}{$n\pi+\pi/2$}{$m\pi-\alpha_3$}{$m\pi- \pi/2$}
\end{center}
and their color-swapped variants.
\end{theorem}

Our method to prove the result is obvious from the previous sections: We only
have to find the Euler equalities satisfied by rational angles to deduce the
equalities satisfied by arbitrary angles.

To find Euler equalities for rational angles, we will use the following result,
original formulated in the context of tilings\footnote{The astute
  reader should remark that this theorem subsumes most of the results
  of Backens \cite{BackensOne} and Mastumoto-Amano \cite{Matsumoto}.}:

\begin{theorem}[\cite{RadinSadun}]  
  \label{thm:radinsadun}
  Let $\alpha_1 \dots \alpha_n \in \mathbb{Q}\pi$ s.t. the diagram\footnotemark
\begin{center}
\begin{tikzpicture}
\begin{pgfonlayer}{nodelayer}
\node [style=none] (0) at (0.5, 2) {};
\node [style=grn,label={[label position=right]$\alpha_n$}] (1) at (0.5, 2.5) {};
\node [style=none] (1.bis) at (0.6, 3) {$\vdots$};
\node [style=rn,label={[label position=right]$\alpha_2$}] (2) at (0.5, 3.5) {};
\node [style=gn,label={[label position=right]$\alpha_1$}] (3) at (0.5, 4) {};
\node [style=none] (4) at (0.5, 4.5) {};
\end{pgfonlayer}
\begin{pgfonlayer}{edgelayer}
\draw[out=270, in= 90] (1.center) to (0.center);
\draw[out=270, in= 90] (2.center) to (1.center);
\draw[out=270, in= 90] (3.center) to (2.center);
\draw[out=270, in= 90] (4.center) to (3.center);
\end{pgfonlayer}
\end{tikzpicture}
\end{center}
represents the identity matrix up to a scalar. Then either there exists $i$
s.t. $\alpha_i \in \{0,\pi\}$ or there exists $i$ s.t. $\{ \alpha_i,
\alpha_{i+1}\} \subseteq \{  \pi/2, -\pi/2\}$.
\end{theorem}  
\footnotetext{the last node is green if $n$ is even, and red otherwise.}

Now we can begin the proof of the theorem. We first investigate the case of
equalities involving  only four nodes:

\begin{proposition}
Suppose that the two diagrams
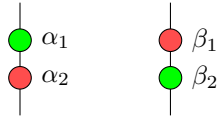
\begin{center}
\begin{tikzpicture}
\begin{pgfonlayer}{nodelayer}
\node [style=none] (0) at (0.5, 1.5) {};
\node [style=rn,label={[label position=right]$\alpha_2$}] (2) at (0.5, 2) {};
\node [style=gn,label={[label position=right]$\alpha_1$}] (3) at (0.5, 2.5) {};
\node [style=none] (4) at (0.5, 3) {};

\node [style=none] (a0) at (2.5, 1.5) {};
\node [style=gn,label={[label position=right]$\beta_2$}] (a2) at (2.5, 2) {};
\node [style=rn,label={[label position=right]$\beta_1$}] (a3) at (2.5, 2.5) {};
\node [style=none] (a4) at (2.5, 3) {};
\end{pgfonlayer}
\begin{pgfonlayer}{edgelayer}
\draw[out=270, in= 90] (2.center) to (0.center);
\draw[out=270, in= 90] (3.center) to (2.center);
\draw[out=270, in= 90] (4.center) to (3.center);
\draw[out=270, in= 90] (a2.center) to (a0.center);
\draw[out=270, in= 90] (a3.center) to (a2.center);
\draw[out=270, in= 90] (a4.center) to (a3.center);
\end{pgfonlayer}
\end{tikzpicture}
\end{center}
represents the same matrix up to a scalar.

Then either ($\alpha_1=\beta_2 = n\pi$ and $\beta_1 = (-1)^{n}
\alpha_2$) or
($\alpha_2=\beta_1 = n\pi$ and $\beta_2 = (-1)^{n}
\alpha_1$).
\end{proposition}
\begin{proof}
  The two diagrams represents respectively the matrices:

  \[
\frac{1}{2}    \begin{pmatrix}
  e^{i\alpha_2} + 1 & (1 - e^{i\alpha_2})e^{i\alpha_1} \\
1 - e^{i\alpha_2} & (1 + e^{i\alpha_2})e^{i\alpha_1} \\    
\end{pmatrix} \]
and
\[
\frac{1}{2}    \begin{pmatrix}
  e^{i\beta_1} + 1 & 1 - e^{i\beta_1}\\
(1 - e^{i\beta_1})e^{i\beta_2}  & (1 + e^{i\beta_1})e^{i\beta_2} \\    
\end{pmatrix}      
\]

\begin{itemize}
\item    Suppose that $\alpha_2 = \pi$. Then we get immediately $\beta_1 =
  \pi$ so that the top-left coefficients are equal.

The quotient of the bottom-left value by the top-right
value should be equal in the two matrices, and we get $\beta_2  = -\alpha_1$.
\item If $\alpha_2 = 0$, we get by a similar argument that $\beta_1 =
  0$ and $\beta_2 = \alpha_1$.
\item  If $\alpha_2 \not \in \{0,\pi\}$, we can look again at the
  two quotients  (as the denominators are nonzero) and we get
  simultaneously $\beta_2 =
  \alpha_1$ and $\beta_2 = -\alpha_1$, and therefore $\beta_2 =
\alpha_1 =  0$ or $\beta_2 = \alpha_1 = \pi$ from which we get the
same conclusions as before by the red-green symmetry.
\end{itemize}
\end{proof}

\begin{proof}[Proof ot the theorem]
It is enough to find all equalities relating rational angles. Take
such an equality. By slightly rewriting it, we see that the diagram
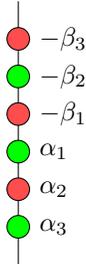
\begin{center}
\begin{tikzpicture}
\begin{pgfonlayer}{nodelayer}
\node [style=none] (0) at (0.5, 1) {};
\node [style=gn,label={[label position=right]$\alpha_3$}] (1) at (0.5, 1.5) {};
\node [style=rn,label={[label position=right]$\alpha_2$}] (2) at (0.5, 2) {};
\node [style=gn,label={[label position=right]$\alpha_1$}] (3) at (0.5, 2.5) {};
\node [style=rn,label={[label position=right]$-\beta_1$}] (a1) at (0.5, 3) {};
\node [style=gn,label={[label position=right]$-\beta_2$}] (a2) at (0.5, 3.5) {};
\node [style=rn,label={[label position=right]$-\beta_3$}] (a3) at (0.5, 4) {};
\node [style=none] (4) at (0.5, 4.5) {};
\end{pgfonlayer}
\begin{pgfonlayer}{edgelayer}
\draw[out=270, in= 90] (1.center) to (0.center);
\draw[out=270, in= 90] (2.center) to (1.center);
\draw[out=270, in= 90] (3.center) to (2.center);
\draw[out=270, in= 90] (a1.center) to (3.center);
\draw[out=270, in= 90] (a2.center) to (a1.center);
\draw[out=270, in= 90] (a3.center) to (a2.center);
\draw[out=270, in= 90] (4.center) to (a3.center);
\end{pgfonlayer}
\end{tikzpicture}
\end{center}
represents the identity matrix upto a scalar.
By Theorem~\ref{thm:radinsadun}, one of the coefficient should be in $\{0,\pi\}$ or
two consecutive coefficients are equal to $\pm \pi/2$.

The rest of the proof is just an enumeration of all possible cases.

\begin{itemize}
\item  $\alpha_1 = n\pi$.
  Then the equation can be rewritten:
\begin{center}
\begin{tikzpicture}
\begin{pgfonlayer}{nodelayer}
\node [style=none] (0) at (1, 1) {};
\node [style=gn,label={[label position=left]$\alpha_3 + n\pi$}] (1) at (1, 1.5) {};
\node [style=rn,label={[label position=left]$(-1)^n \alpha_2$}] (2) at (1, 2.5) {};
\node [style=none] (4) at (1, 3) {};
\node [style=none] (eq) at (1.5, 2) {=};
\node [style=none] (a0) at (2, 1) {};
\node [style=rn,label={[label position=right]$\beta_3$}] (a1) at (2, 1.5) {};
\node [style=gn,label={[label position=right]$\beta_2$}] (a2) at (2, 2) {};
\node [style=rn,label={[label position=right]$\beta_1$}] (a3) at (2, 2.5) {};
\node [style=none] (a4) at (2, 3) {};
\end{pgfonlayer}
\begin{pgfonlayer}{edgelayer}
\draw[out=270, in= 90] (1.center) to (0.center);
\draw[out=270, in= 90] (2.center) to (1.center);
\draw[out=270, in= 90] (4.center) to (2.center);
\draw[out=270, in= 90] (a1.center) to (a0.center);
\draw[out=270, in= 90] (a2.center) to (a1.center);
\draw[out=270, in= 90] (a3.center) to (a2.center);
\draw[out=270, in= 90] (a4.center) to (a3.center);
\end{pgfonlayer}
\end{tikzpicture}
\end{center}
and therefore

\begin{center}
\begin{tikzpicture}
\begin{pgfonlayer}{nodelayer}
\node [style=none] (0) at (1, 1) {};
\node [style=gn,label={[label position=left]$\alpha_3+n\pi$}] (2) at (1, 1.5) {};
\node [style=rn,label={[label position=left]$(-1)^n \alpha_2-\beta_1$}] (3) at (1, 2.5) {};
\node [style=none] (4) at (1, 3) {};
\node [style=none] (eq) at (1.5, 2) {=};

\node [style=none] (a0) at (2, 1) {};
\node [style=rn,label={[label position=right]$\beta_3$}] (a2) at (2, 1.5) {};
\node [style=gn,label={[label position=right]$\beta_2$}] (a3) at (2, 2.5) {};
\node [style=none] (a4) at (2, 3) {};
\end{pgfonlayer}
\begin{pgfonlayer}{edgelayer}
\draw[out=270, in= 90] (2.center) to (0.center);
\draw[out=270, in= 90] (3.center) to (2.center);
\draw[out=270, in= 90] (4.center) to (3.center);
\draw[out=270, in= 90] (a2.center) to (a0.center);
\draw[out=270, in= 90] (a3.center) to (a2.center);
\draw[out=270, in= 90] (a4.center) to (a3.center);
\end{pgfonlayer}
\end{tikzpicture}
\end{center}
upto a scalar. We therefore get the first two cases of the theorem by the previous proposition.

\item The cases $\alpha_3 = n\pi, \beta_1 = n\pi, \beta_3= n\pi$ are
  treated in the same way and give symmetric conditions.

\item The cases $\alpha_2 = n\pi, \beta_2 = n\pi$ can be treated by rewriting the equality in the form
  
\begin{center}
\begin{tikzpicture}
\begin{pgfonlayer}{nodelayer}
\node [style=none] (0) at (1, 1) {};
\node [style=rn,label={[label position=left]$-\beta_3$}] (1) at (1, 1.5) {};
\node [style=gn,label={[label position=left]$\alpha_3$}] (2) at (1, 2) {};
\node [style=rn,label={[label position=left]$\alpha_2$}] (3) at (1, 2.5) {};
\node [style=none] (4) at (1, 3) {};
\node [style=none] (eq) at (1.5, 2) {=};
\node [style=none] (a0) at (2, 1) {};
\node [style=gn,label={[label position=right]$\beta_2$}] (a1) at (2, 1.5) {};
\node [style=rn,label={[label position=right]$\beta_1$}] (a2) at (2, 2) {};
\node [style=gn,label={[label position=right]$-\alpha_1$}] (a3) at (2, 2.5) {};
\node [style=none] (a4) at (2, 3) {};
\end{pgfonlayer}
\begin{pgfonlayer}{edgelayer}
\draw[out=270, in= 90] (1.center) to (0.center);
\draw[out=270, in= 90] (2.center) to (1.center);
\draw[out=270, in= 90] (3.center) to (2.center);
\draw[out=270, in= 90] (4.center) to (3.center);
\draw[out=270, in= 90] (a1.center) to (a0.center);
\draw[out=270, in= 90] (a2.center) to (a1.center);
\draw[out=270, in= 90] (a3.center) to (a2.center);
\draw[out=270, in= 90] (a4.center) to (a3.center);
\end{pgfonlayer}
\end{tikzpicture}
\end{center}
and we recover one of the previous cases.
\item  $\alpha_1 = \pi/2$ and $\alpha_2 = \pi/2$.
  Then the equation can be rewritten:

\begin{center}
\begin{tikzpicture}
\begin{pgfonlayer}{nodelayer}
\node [style=none] (0) at (1, 1) {};
\node [style=gn,label={[label position=left]$\alpha_3-\pi/2$}] (1) at (1, 1.75) {};
\node [style={H box}] (2) at (1, 2.5) {};
\node [style=none] (4) at (1, 3) {};
\node [style=none] (eq) at (1.5, 2) {=};
\node [style=none] (a0) at (2, 1) {};
\node [style=rn,label={[label position=right]$\beta_3$}] (a1) at (2, 1.5) {};
\node [style=gn,label={[label position=right]$\beta_2$}] (a2) at (2, 2) {};
\node [style=rn,label={[label position=right]$\beta_1$}] (a3) at (2, 2.5) {};
\node [style=none] (a4) at (2, 3) {};
\end{pgfonlayer}
\begin{pgfonlayer}{edgelayer}
\draw[out=270, in= 90] (1.center) to (0.center);
\draw[out=270, in= 90] (2.center) to (1.center);
\draw[out=270, in= 90] (4.center) to (2.center);
\draw[out=270, in= 90] (a1.center) to (a0.center);
\draw[out=270, in= 90] (a2.center) to (a1.center);
\draw[out=270, in= 90] (a3.center) to (a2.center);
\draw[out=270, in= 90] (a4.center) to (a3.center);
\end{pgfonlayer}
\end{tikzpicture}
\end{center}
and therefore
\begin{center}
\begin{tikzpicture}
\begin{pgfonlayer}{nodelayer}
\node [style=none] (0) at (1, 1) {};
\node [style=rn,label={[label position=left]$\alpha_3-\pi/2$}] (2) at (1, 2.5) {};
\node [style={H box}] (1) at (1, 1.75) {};
\node [style=none] (4) at (1, 3) {};
\node [style=none] (eq) at (1.5, 2) {=};

\node [style=none] (a0) at (2, 1) {};
\node [style=rn,label={[label position=right]$\beta_3$}] (a1) at (2, 1.5) {};
\node [style=gn,label={[label position=right]$\beta_2$}] (a2) at (2, 2) {};
\node [style=rn,label={[label position=right]$\beta_1$}] (a3) at (2, 2.5) {};
\node [style=none] (a4) at (2, 3) {};
\end{pgfonlayer}
\begin{pgfonlayer}{edgelayer}
\draw[out=270, in= 90] (1.center) to (0.center);
\draw[out=270, in= 90] (2.center) to (1.center);
\draw[out=270, in= 90] (4.center) to (2.center);
\draw[out=270, in= 90] (a1.center) to (a0.center);
\draw[out=270, in= 90] (a2.center) to (a1.center);
\draw[out=270, in= 90] (a3.center) to (a2.center);
\draw[out=270, in= 90] (a4.center) to (a3.center);
\end{pgfonlayer}
\end{tikzpicture}
\end{center}
which gives
\begin{center}
\begin{tikzpicture}
\begin{pgfonlayer}{nodelayer}
\node [style=none] (0) at (1, 1) {};
\node [style=rn,label={[label position=left]$\pi/2$}] (1) at (1, 1.5) {};
\node [style=gn,label={[label position=left]$\pi/2$}] (2) at (1, 2) {};
\node [style=rn,label={[label position=left]$\alpha_3$}] (3) at (1, 2.5) {};
\node [style=none] (4) at (1, 3) {};
\node [style=none] (eq) at (1.5, 2) {=};

\node [style=none] (a0) at (2, 1) {};
\node [style=rn,label={[label position=right]$\beta_3$}] (a1) at (2, 1.5) {};
\node [style=gn,label={[label position=right]$\beta_2$}] (a2) at (2, 2) {};
\node [style=rn,label={[label position=right]$\beta_1$}] (a3) at (2, 2.5) {};
\node [style=none] (a4) at (2, 3) {};
\end{pgfonlayer}
\begin{pgfonlayer}{edgelayer}
\draw[out=270, in= 90] (1.center) to (0.center);
\draw[out=270, in= 90] (2.center) to (1.center);
\draw[out=270, in= 90] (4.center) to (3.center);
\draw[out=270, in= 90] (3.center) to (2.center);
\draw[out=270, in= 90] (a1.center) to (a0.center);
\draw[out=270, in= 90] (a2.center) to (a1.center);
\draw[out=270, in= 90] (a3.center) to (a2.center);
\draw[out=270, in= 90] (a4.center) to (a3.center);
\end{pgfonlayer}
\end{tikzpicture}
\end{center}
Finally:
\begin{center}
\begin{tikzpicture}
\begin{pgfonlayer}{nodelayer}
\node [style=none] (0) at (1, 1.5) {};
\node [style=gn,label={[label position=left]$\pi/2$}] (2) at (1, 2) {};
\node [style=rn,label={[label position=left]$\alpha_3-\beta_1$}] (3) at (1, 3) {};
\node [style=none] (4) at (1, 3.5) {};
\node [style=none] (eq) at (1.5, 2) {=};

\node [style=none] (a0) at (2, 1.5) {};
\node [style=rn,label={[label position=right]$\beta_3-\pi/2$}] (a1) at (2, 2) {};
\node [style=gn,label={[label position=right]$\beta_2$}] (a2) at (2, 3) {};
\node [style=none] (a4) at (2, 3.5) {};
\end{pgfonlayer}
\begin{pgfonlayer}{edgelayer}
\draw[out=270, in= 90] (2.center) to (0.center);
\draw[out=270, in= 90] (4.center) to (3.center);
\draw[out=270, in= 90] (3.center) to (2.center);
\draw[out=270, in= 90] (a1.center) to (a0.center);
\draw[out=270, in= 90] (a2.center) to (a1.center);
\draw[out=270, in= 90] (a4.center) to (a2.center);
\end{pgfonlayer}
\end{tikzpicture}
\end{center}
and we can apply the previous proposition. We obtain the fourth case
of the theorem for $n = 0$ (notice that $(-1)^m \pi/2 =
m\pi+\pi/2$).
\item $\alpha_1 = -\pi/2$ and $\alpha_2 = -\pi/2$ is treated in a
  similar way and gives use the fourth case with $n  = 1$.

\item $\alpha_1 = -\pi/2$ and $\alpha_2 = \pi/2$.
The same computation gives us:
\begin{center}
\begin{tikzpicture}
\begin{pgfonlayer}{nodelayer}
\node [style=none] (0) at (1, 1) {};
\node [style=rn,label={[label position=left]$\pi/2$}] (1) at (1, 1.5) {};
\node [style=gn,label={[label position=left]$\pi/2$}] (2) at (1, 2) {};
\node [style=rn,label={[label position=left]$\alpha_3$}] (3) at (1, 2.5) {};
\node [style=gn,label={[label position=left]$\pi$}] (4) at (1, 3) {};
\node [style=none] (5) at (1, 3.5) {};
\node [style=none] (eq) at (1.5, 2) {=};
\node [style=none] (a0) at (2, 1) {};
\node [style=rn,label={[label position=right]$\beta_3$}] (a1) at (2, 1.5) {};
\node [style=gn,label={[label position=right]$\beta_2$}] (a2) at (2, 2) {};
\node [style=rn,label={[label position=right]$\beta_1$}] (a3) at (2, 2.5) {};
\node [style=none] (a4) at (2, 3) {};
\end{pgfonlayer}
\begin{pgfonlayer}{edgelayer}
\draw[out=270, in= 90] (1.center) to (0.center);
\draw[out=270, in= 90] (2.center) to (1.center);
\draw[out=270, in= 90] (4.center) to (3.center);
\draw[out=270, in= 90] (5.center) to (4.center);
\draw[out=270, in= 90] (3.center) to (2.center);
\draw[out=270, in= 90] (a1.center) to (a0.center);
\draw[out=270, in= 90] (a2.center) to (a1.center);
\draw[out=270, in= 90] (a3.center) to (a2.center);
\draw[out=270, in= 90] (a4.center) to (a3.center);
\end{pgfonlayer}
\end{tikzpicture}
\end{center}
and finally:
\begin{center}
\begin{tikzpicture}
\begin{pgfonlayer}{nodelayer}
\node [style=none] (0) at (1, 1.5) {};
\node [style=gn,label={[label position=left]$-\pi/2$}] (2) at (1, 2) {};
\node [style=rn,label={[label position=left]$-\alpha_3-\beta_1$}] (3) at (1, 3) {};
\node [style=none] (4) at (1, 3.5) {};
\node [style=none] (eq) at (1.5, 2) {=};

\node [style=none] (a0) at (2, 1.5) {};
\node [style=rn,label={[label position=right]$\beta_3-\pi/2$}] (a1) at (2, 2) {};
\node [style=gn,label={[label position=right]$\beta_2$}] (a2) at (2, 3) {};
\node [style=none] (a4) at (2, 3.5) {};
\end{pgfonlayer}
\begin{pgfonlayer}{edgelayer}
\draw[out=270, in= 90] (2.center) to (0.center);
\draw[out=270, in= 90] (4.center) to (3.center);
\draw[out=270, in= 90] (3.center) to (2.center);
\draw[out=270, in= 90] (a1.center) to (a0.center);
\draw[out=270, in= 90] (a2.center) to (a1.center);
\draw[out=270, in= 90] (a4.center) to (a2.center);
\end{pgfonlayer}
\end{tikzpicture}
\end{center}
And we apply again the proposition.
\item The case $\{\alpha_2, \alpha_3\} \subseteq \{ \pi/2, -\pi/2\}$
  is done in the same way.
  The case $\{\alpha_1, \beta_1\} \subseteq \{ \pi/2, -\pi/2\}$ is
  done by rewriting the equality as before.
\end{itemize}  

\end{proof}

\section{Further work}

The work that has been done for the Euler equalities can be also done
for the rule $(A)$ itself, to find for which angles rule $(A)$ is
actually provable without rule $(A)$. The technical condition below
rule $(A)$ can be rewritten as a sum of 6 exponential terms, and
has been investigated previously by Mann \cite{Mann} and Conway and
Jones\cite{CJ} (see in particular Theorem 6).
The most complicated equality with rational angles uses angles of the
form $\pi/3$ and $\pi/5$.

Replacing the rule $CYC_p$ (or the metarule) by the rule $SUP_p$ gives
a nicer set of rules for the rational fragment of the ZX-calculus.
What remains is to find a nicer replacement for rule $(A)$ for the
full fragment.
The obvious candidate is of course Euler equalities, but we do not
know if adding Euler equalities to the rule of $ZX$ is sufficient to
prove rule $(A)$ in its full generality.


\end{document}